\documentclass[runningheads]{llncs}

\usepackage{amsmath,amssymb, amsthm}
\usepackage[normalem]{ulem}
\usepackage{bm}
\usepackage{tikz}
\usepackage{hyperref}
\usepackage{ltl}
\usepackage{graphicx}

\usepackage{subcaption} 
\usepackage[strings]{underscore}	 
\usepackage{pifont}
\usepackage{cleveref}
\usepackage{xspace}
\usepackage[ruled,vlined]{algorithm2e} 
\usepackage{mathtools}
\usetikzlibrary{arrows,automata,positioning,shapes,backgrounds,patterns,calc}
\usepackage{subcaption}
\usepackage{xcolor}
\usepackage{multirow}
\usepackage{cite}

\newcommand\ldot{\mathpunct{.}}

\newcommand{\Lang}{\mathcal L}
\newcommand{\A}{\mathcal A}
\newcommand{\variables}{\mathcal V}
\newcommand{\informedness}{IC}

\renewcommand{\models}{\vDash}

\newcommand{\donotshow}[1]{}

\newcommand{\vc}{\mathtt{c}}
\newcommand{\vin}{\mathtt{in}}
\newcommand{\vout}{\mathtt{out}}

\newcommand{\vt}{\mathtt{t}}

\newcommand{\vic}{\mathtt{ic}}
\newcommand{\bosyhyper}{\textsc{BoSyHyper}\xspace}

\theoremstyle{plain}
\newtheorem*{theorem*}{Theorem}
\newtheorem*{lemma*}{Lemma}

\newcommand{\twopartdef}[4]
{
	\left\{
		\begin{array}{ll}
			#1 & \mbox{if } #2 \\
			#3 & \mbox{if } #4
		\end{array}
	\right.
}

\begin{document}

\title{Information Flow Guided Synthesis\\ (Full Version)\thanks{This work was funded by the German Israeli Foundation (GIF) Grant No. I-1513-407./2019. and by the DFG grant 389792660 as part of \href{https://perspicuous-computing.science}{TRR~248 -- CPEC}.}}
\titlerunning{Information Flow Guided Synthesis}

\author{}
\institute{}
\authorrunning{B. Finkbeiner and N. Metzger and Y. Moses}
\author{Bernd Finkbeiner\inst{1}
\and 
Niklas Metzger\inst{1}
\and 
Yoram Moses\inst{2}}

\institute{CISPA Helmholtz Center for Information Security, Saarland, Germany \email{\{finkbeiner, niklas.metzger\}@cispa.de} 
\and
The Andrew and Erna Viterbi Faculty of Electrical and Computer\\ Engineering and the Taub Faculty of Computer Science, Technion, Israel\\
\email{moses@ee.technion.ac.il}
}

\maketitle

\begin{abstract}

Compositional synthesis relies on the discovery of assumptions, i.e., restrictions on the behavior of the remainder of the system that allow a component to realize its specification.
In order to avoid losing valid solutions, these assumptions should be \emph{necessary} conditions for realizability. However, because there are typically many different behaviors that realize the same specification,  necessary behavioral restrictions often do not exist.
In this paper, we introduce a new class of assumptions for compositional synthesis, which we call \emph{information flow assumptions}. Such assumptions capture an essential aspect of distributed computing, because components often need to act upon information that is available only in other components. The presence of a certain flow of information is therefore often a necessary requirement, while
the actual behavior that establishes the information flow is  unconstrained. 
In contrast to behavioral assumptions, which are properties of individual computation traces, information flow assumptions are \emph{hyperproperties}, i.e., properties of sets of traces. We present a method for the automatic derivation of information-flow assumptions from a temporal logic specification of the system. We then provide a technique for the automatic synthesis of component implementations based on information flow assumptions. This provides a new compositional approach to the synthesis of distributed systems. We report on encouraging first experiments with the approach, carried out with the \bosyhyper synthesis tool.

\end{abstract}

\section{Introduction}\label{sec:introduction}
In \emph{distributed synthesis}, we are interested in the automatic translation of a formal specification of a   distributed system's desired behavior into an  implementation that satisfies the specification~\cite{PnueliR90}. What makes distributed synthesis far more  interesting than the standard synthesis of reactive systems, but also more challenging, is that the result consists of a set of implementations of subsystems, each of which operates based only on  partial knowledge of the global system state. 
While algorithms for distributed synthesis have been studied since
the 1990s~\cite{PnueliR90,kv01,fs05}, their high complexity has resulted in applications of distributed synthesis being, so far, very limited.

One of the most promising approaches to making distributed synthesis more scalable is \emph{compositional synthesis}~\cite{DBLP:journals/ijfcs/ScheweF07,KupfermanPV06,FiliotJR10,DammF14,FinkbeinerP20}. 
The compositional synthesis of a distributed system with two processes, $p$ and $q$, avoids the construction of the product of $p$ and $q$ and instead focuses on one process at a time.
Typically, it is impossible to realize one process without making certain assumptions about the other process.
Compositional  synthesis  therefore  critically  depends  on  finding  the  assumption that $p$ must make about $q$, and vice versa: once the assumptions are known, one can build each individual process, relying on  
the fact that the assumption will be satisfied by the synthesized implementation of the other process. Ideally, the assumptions should be both \emph{sufficient} (i.e., the processes are realizable under the assumptions) and \emph{necessary} (i.e., any implementation that satisfies the specification would also satisfy the assumptions). Without sufficiency, the synthesis cannot find a compositional solution; without necessity, the synthesis loses valid solutions. While sufficiency is obviously checked as part of the synthesis process, it is often impossible to find necessary conditions, because there the specifications can be realized by many different behaviors. Any specific implementation would lead to a specific assumption; however, this implementation is only known once the synthesis is complete, and an assumption that is satisfied by \emph{all} implementations often does not exist.

In this paper, we propose a way out of this chicken-and-egg type of situation. 
Previous work on generating assumptions for compositional synthesis has focused on \emph{behavioral} restrictions on the environment of a subsystem.  We introduce a new class of more abstract assumptions that, instead, focus on  
the \emph{flow of information}.
Consider a system architecture (depicted in  \Cref{fig:architecture}) where two processes $a$
and~$b$ are linked by a communication channel $\vc$, such that $a$ can write to~$\vc$ and $b$ can read from $\vc$. Suppose also that $a$ reads a boolean input $\vin$ from the environment that is, however, not directly visible to $b$. 
We are interested in a distributed implementation for a specification that demands that $b$ should eventually output the value of input $\vin$. Since $b$ cannot observe $\vin$, its synthesis must rely on the assumption that the value of $\vin$ is communicated over the channel $\vc$ by process $a$.  Expressing this as a \emph{behavioral assumption} is difficult, because there are many different behaviors that accomplish this. 
Process $a$ could, for example, literally copy the value of $\vin$ to $\vc$. 
It could also encode the value, for example by writing to~$\vc$ the negation of the value of $\vin$.  
Alternatively, it could delay the transmission of $\vin$ by an arbitrary number of steps, and even use the length of the delay to encode information about the value of~$\vin$. 
Fixing any such communication protocol, by a corresponding behavioral assumption on $a$, would unnecessarily eliminate potential implementations of $b$.
The minimal assumption that subsystem $a$ must rely  
on is in fact an information-flow assumption, namely that $b$ will eventually learn the value of $\vin$.

We present a method that derives necessary information flow assumptions automatically.  
A fundamental difference between behavioral and information flow assumptions is that behavioral assumptions are  \emph{trace properties}, i.e., properties of individual traces; by contrast, information flow assumptions are \emph{hyperproperties}, i.e., properties of \emph{sets} of traces. 
In our example, the assumption that $a$ will eventually communicate the value of $\vin$ to $b$ is the hyperproperty that any two traces that differ in the value of~$\vin$ must eventually also differ in~$\vc$. 
The precise difference between the two traces depends on the communication protocol chosen in the implementation of $a$; however, any correct implementation of $a$ must ensure that some difference 
in~$b$'s input (on channel~$\vc$) in the two traces 
occurs, so that $b$ can then respond with a different output. 

Once we have obtained information flow assumptions for all of the subsystems, we proceed to synthesize each subsystem under the assumption generated for its environment. 
It is important to note that, at this point, the implementation of the environment is not known yet; as a result, we only know \emph{what} information will be provided to process~$b$, 
but not \emph{how}. 
This also means that we cannot yet construct an executable implementation of the process under consideration; after all, this implementation would need to correctly decode the information provided by its partner process.
Clearly, we cannot determine how to \emph{decode} the information before we know how the implementation of the sending process \emph{encodes} the information!

Our solution to this quandary is to synthesize a prototype of an implementation for the process that works with \emph{any} implementation of the sender, as long as the sender satisfies the information flow requirement. 
The prototype differs from the actual implementation in that it has access to the original (unencoded) information. Because of this information the prototype, which we call a {\it hyper implementation}, can determine the correct output that satisfies the specification. Later, in the actual implementation, the information is no longer available in its original, unencoded form, but must instead be decoded from the communication received from the environment.
However, the information flow assumption guarantees that this is actually possible, and access to the original information is, therefore, no longer necessary.

In Section~\ref{sec:bit:trans}, we explain our approach in more detail, continuing the discussion of the bit transmission example mentioned above. The paper then proceeds to make the following contributions:
\begin{itemize}
    \item We introduce the notion of \emph{necessary information flow assumptions} (Section \ref{sec:Necessary:Information:Flow:Assumptions}) for distributed systems with two processes and present a method for the automatic derivation of such assumptions from process specifications given in linear-time temporal logic (LTL).
    \item We strengthen information flow assumptions to the notion of \emph{time-bounded} information flow assumptions (Section \ref{sec:Time:Bounded:Information:Flow}), which characterizes information that must be received in finite time. We introduce the notion of \emph{uniform distinguishability} and prove that uniform distinguishability guarantees the necessity of the information flow assumption.
    \item We introduce the notion of \emph{hyper implementations} (Section~\ref{sec:behaviour}) and provide a synthesis method for their automatic construction. We also explain how to transform hyper implementations into actual process implementations.
    \item We present a \emph{practical approach} (Section~\ref{sec:practical}) that simplifies the synthesis for cases where the information flow assumption refers to a finite amount of information.
    \item We report on encouraging experimental results (Section~\ref{sec:experiments}). 
\end{itemize}

\section{The Bit Transmission Problem}\label{sec:bit:trans}
We use the \textit{bit transmission} example from the introduction to motivate our approach. The example consists of two processes $a$ and $b$ that are combined into the distributed architecture shown in \Cref{fig:architecture}.
Process $a$ observes the (binary) input of the environment through variable $\vin$ and can communicate with the second process $b$ via a channel (modeled by the shared variable~$\vc$).
Process $b$ observes its own local input from $a$ and has a local output $\vout$.
We are interested in synthesizing an implementation for our distributed system consisting of two strategies, one for each process, whose  
combined behavior satisfies the specification. In this example, the specification for process~$b$ is to transmit the initial value of~$\vin$, an input of~$a$,  to $b$'s own output; this is expressed by the linear-time temporal logic (LTL) formula $\varphi_{b} =\vin \LTLequ \LTLeventually \vout$. 
The specification does not restrict~$a$'s behavior, such that $\varphi_{a}=\mathit{true}$.

Since the value of~$\vout$ is controlled by $b$, whereas~$\vin$ is determined by the environment and observed by $a$, this specification forces $b$ to react to an input that~$b$  neither observes nor controls. 
To satisfy the goal, $\vout$ must remain $\mathit{false}$ forever if~$\vin$ is initially $\mathit{false}$, while $\vout$ must eventually become $\mathit{true}$ at least once if $\vin$ starts with value~$\mathit{true}$.
Indeed, in order to set~$\vout$ to $\mathit{true}$, process~$b$ must \emph{know} that~$\vin$ is initially $\mathit{true}$, which can only be satisfied via information flow from $a$ to $b$.
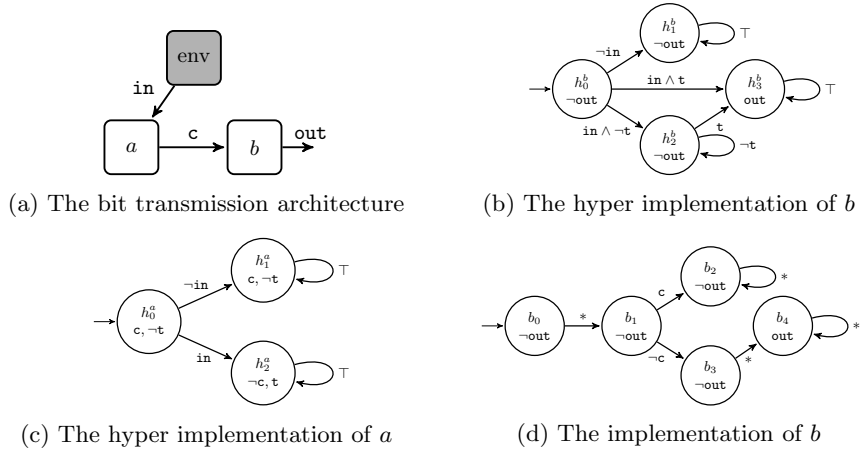
\begin{figure}[t]
     \centering
     \begin{subfigure}[b]{0.46\textwidth}
         \centering
         \resizebox{.85\textwidth}{!}{
         \tikzstyle{state}=[draw, rectangle, fill=none, minimum width=.8cm, 
minimum height = .8cm, rounded corners=1mm, align=center, thick]

\begin{tikzpicture}[->,>=stealth',shorten >= 1pt,auto]
\node[state] (p1) {%
  $~a$
};
\node (left) [left =0.7 of p1]{};
\node[state] (p2) [right= 1 of p1] {%
  $b$};
\node[state, fill=gray!60] (env) [above right = .5 and .1 of p1]{$\text{env}$};
\node[state, draw = none, minimum height=.1cm] (p2below) [right= .5 of p2] {%
 };

\path (env) edge[thick, transform canvas={yshift=0mm}] node[above left] {%
        $\vin$
      }  (p1)
      (p1) edge[thick, transform canvas={yshift=0mm}] node[above] {%
        $\vc$
      } (p2)
      (p2) edge[thick, transform canvas={xshift=0mm}] node[above right, xshift=-2mm] {
      $\vout$
      }(p2below)
      ;
\end{tikzpicture}
         }
         \caption{The bit transmission architecture}
         \label{fig:architecture}
     \end{subfigure}
     \begin{subfigure}[b]{0.53\textwidth}
         \centering
         \resizebox{.75\textwidth}{!}{
         \tikzstyle{state}=[draw, circle, fill=none, minimum width=1cm, 
minimum height = 1.3cm, align=center, thick]

\begin{tikzpicture}[->,>=stealth',shorten >= 1pt,auto]

\node[state] (p0) {%
    $h_0^b$\\
   $\neg \vout$
};

\node (left) [left = 0.5 of p0]{};

\node[state] (p1) [above right = .3 and 1 of p0] {%
  $h_1^b$\\
  $\neg \vout$
 };
\node[state] (p2) [below right= .3 and 1 of p0]{
    $h_2^b$\\
    $\neg \vout$
    };

\node[state] (p3) [right = 2.5 of p0] {%
    $h_3^b$\\
    $\vout$
  };

\path (left) edge (p0)
      (p0) edge[thick] node[above left] {%
        $\neg \vin$
      } (p1)
      (p0) edge[thick] node [above, align=center] {%
      $ \vin \wedge \vt$
      }(p3)
      (p0) edge[thick] node [below left, xshift=10pt,yshift=-2pt, align=center] {%
      $ \vin \wedge \neg \vt$
      }(p2)
       (p0) edge[loop left, thick, draw=none] node [left] {%
       $\phantom{\LTLtrue}$}
       (p0) 
     (p2) edge[thick] node[below right, xshift=-1pt] {
      $\vt$
      }(p3)
      (p3) edge[loop right, thick] node [right] {%
      $\LTLtrue$}(p3) 
      (p2) edge[loop right, thick] node [right] {%
      $\neg \vt$}(p2) 
      (p1) edge[loop right, thick] node [right] {%
      $\LTLtrue$}(p1)
      ;
\end{tikzpicture}
         }
         \caption{The hyper implementation of $b$}
         \label{fig:hyper-implementation-b}
     \end{subfigure}
     \begin{subfigure}[b]{0.46\textwidth}
         \centering
         \resizebox{.75\linewidth}{!}{
         \tikzstyle{state}=[draw, circle, fill=none, minimum width=1cm, 
minimum height = 1.4cm, align=center, thick]
\begin{tikzpicture}[->,>=stealth',shorten >= 1pt,auto]

\node[state] (p0) {%
    $h^a_0$\\
   $\vc, \neg \vt$
};

\node (left) [left = 0.5 of p0]{};
\node (above) [above = 0.05 of p1]{};

\node[state] (p1) [above right = .12 and 1.5 of p0] {%
  $h^a_1$\\
  $\vc, \neg \vt$
 };
\node[state] (p2) [below right= .12 and 1.5 of p0]{
    $h^a_2$\\
    $\neg \vc, \vt$
    };


\path (left) edge (p0)
      (p0) edge[thick] node[above left, xshift=5pt, yshift=1pt] {%
        $\neg \vin$
      } (p1)
      (p0) edge[thick] node [below, align=center] {%
      $ \vin$
      }(p2)
       (p0) edge[loop left, thick, draw=none] node [left] {%
       $\phantom{\LTLtrue}$}
       (p0) 
      (p2) edge[loop right, thick] node [right] {%
      $\LTLtrue$}(p2) 
      (p1) edge[loop right, thick] node [right] {%
      $\LTLtrue$}(p1)
      ;
\end{tikzpicture}
         }
         \caption{The hyper implementation of $a$}
         \label{fig:implemetation of -a}
     \end{subfigure}
    \begin{subfigure}[b]{0.53\textwidth}
         \centering
         \resizebox{.85\linewidth}{!}{
         \tikzstyle{state}=[draw, circle, fill=none, minimum width=1cm, 
minimum height = 1.3cm, align=center, thick]

\begin{tikzpicture}[->,>=stealth',shorten >= 1pt,auto]

\node[state] (p0){%
$b_0$\\
$\neg \vout$
};
\node[state] (p1) [right = .8 of p0]{%
    $b_1$\\
   $\neg \vout$
};

\node (left) [left = 0.5 of p0]{};

\node[state] (p2) [above right = .15 and .8 of p1] {%
  $b_2$\\
  $\neg \vout$
 };
\node[state] (p3) [below right= .15 and .8 of p1]{
    $b_3$\\
    $\neg \vout$
    };

\node[state] (p4) [right = 2 of p1] {%
    $b_4$\\
    $\vout$
  };

\path (left) edge (p0)
(p0) edge[thick] node[above] {%
        $\ast$
      } (p1)
      (p1) edge[thick] node[above left] {%
        $\vc$
      } (p2)
      (p1) edge[thick] node [below left,align=center, xshift={1mm}] {%
      $ \neg \vc$
      }(p3)
       (p1) edge[loop left, thick, draw=none] node [left] {%
       $\phantom{\ast}$}
       (p1) 
     (p3) edge[thick] node[below, xshift=1pt] {
      $\ast$
      }(p4)
      (p4) edge[loop right, thick] node [right] {%
      $\ast$}(p4) 
      (p2) edge[loop right, thick] node [right] {%
      $\ast$}(p2)
      ;
\end{tikzpicture}
         }
         \caption{The implementation of $b$}
         \label{fig:implementation-b}
     \end{subfigure}
        \caption{The distributed system of the \emph{bit transmission} protocol. The architecture is given in (a), the hyper-implementation of $b$ in (b), the local implementation of $a$ in(c), and the resulting local implementation of $b$ in (d).}
        \label{fig:bit-transmission}
\end{figure}
We can capture this information flow requirement as the following hyperproperty:
For every pair of traces that disagree on the initial value of~$\vin$, process~$a$ must (eventually) behave differently on $\vc$.
The requirement can be expressed  in HyperLTL by the formula 
$\Psi=\forall \pi,\pi'. (\vin_{\pi} \nleftrightarrow {\vin}_{\pi'}) \rightarrow \LTLeventually ({\vc}_{\pi} \nleftrightarrow {\vc}_{\pi'})$.
The information flow requirement does not restrict $a$ to behave in a particular manner; the \emph{encoding} of the information about $\vin$ on the channel $\vc$ depends on $a$'s behavior.
Under the assumption that~$a$ will behave according to the information flow requirement~$\Psi$, one can synthesize a solution of $b$ that is correct for every implementation of $a$.
Given its generality, we call such a solution a \emph{hyper implementation}, shown in  \Cref{fig:hyper-implementation-b}.
Since the point in time when the information is received by $b$ is unknown during the local synthesis process, an additional boolean variable $\vt$ is added to the specification of $b$.
This variable signals that the information has been transmitted and is later derived by $a$'s implementation.
Setting $\vout$ to \textit{true} is only allowed after $\vt$ is observed by process~$b$.
When the hyper implementation is composed with the actual implementation of~$a$, as shown in \Cref{fig:implemetation of -a}, both local specifications are satisfied.
The resulting local implementation of $b$, depicted in \Cref{fig:implementation-b}, branches only on local inputs and, together with $a$, satisfies the specification.
While changing state $b_0$ to $b_1$, process~$b$ cannot distinguish $\vin$ from $\neg \vin$. It has to wait for one time step, i.e., the first difference in outputs of process $a$, to observe the difference in the shared communication channel.
The value of $\vt$ is obtained from $a$'s implementation and set to $\mathit{true}$ with the first difference in $\vc$, forbidding the edge from $h_0^b$ to $h_3^b$in the local implementation of $b$.

\section{Preliminaries}
\paragraph{Architectures.} For ease of exposition we focus in this paper on systems with two processes. 
Let $\variables$ be a set of variables. An architecture with two black-box processes $p$ and $q$ is given as a tuple $(I_p,I_q,O_p,O_q,I_e)$, where $I_p,I_q,O_p,O_q,$ and $I_e$ are all subsets of $\variables$. $O_p$ and $O_q$ are the \emph{output variables} of $p$ and $q$. $O_e$ are the output variables of the uncontrollable environment. The three sets $O_p, O_q$ and~$O_e$ form a partition of $\variables$.
$I_p$ and $I_q$ are the \emph{input variables} of processes $p$ and~$q$, respectively. 
For each black-box process, the inputs and outputs are disjoint, i.e., $I_p \cap O_p = \emptyset$ and $I_q \cap O_q = \emptyset$. The inputs $I_p$ and~$I_q$ of the black-box processes are all either outputs of the environment or outputs of the other black-box process, i.e., $I_p \subseteq O_q \cup O_e$ and $I_q \subseteq O_p \cup O_e$. 
We assume that all variables are of boolean type.
For a set $V \subseteq \variables$, every subset $V' \subseteq V$ defines a \emph{valuation} of $V$, where
the variables in $V \cap V'$ have value $\mathit{true}$ and the variables in $V \setminus V'$ have value $\mathit{false}$.

\paragraph{Implementations.} An implementation of an architecture $(I_p,I_q,O_p,O_q,I_e)$ is a pair $(s_p, s_q)$, consisting of a strategy for each of the two black-box processes. 
A \emph{strategy} for a black-box process $p$ is a function
$s_p : (2^{I_p})^* \rightarrow (2^{O_p})$ that maps finite sequences of valuations of $p$'s input variables (i.e., \emph{histories} of inputs) to a valuation of $p$'s output variables. 
The (synchronous) \emph{composition}  $s_p || s_q$ of the two strategies is the function $s: (2^{O_e})^* \rightarrow (2^{\variables})$ that maps finite sequences of valuations of the environment's output variables to valuations of all variables: we define $s(\epsilon) = s_p(\epsilon) \cup s_q(\epsilon)$ and, for $v \in (2^{O_e})^*, x \in 2^{O_e}$, $s(v \cdot x) = (s_p(f_p(v)) \cup s_q(f_q(v)) \cup x)$, where $f_p$ and $f_q$ map sequences of environment outputs to sequences of process inputs with $f_p(\epsilon) = \epsilon,  f_p(v \cdot x) = f_p(v) \cdot ((x \cup s_q(f_q(v))) \cap I_p)$ and $f_q(\epsilon) = \epsilon,  f_q(v \cdot x) = f_p(v) \cdot ((x \cup s_p(f_p(v))) \cap I_q)$.

\paragraph{Specifications.}
Our specifications refer to traces over the set $\variables$ of all variables. In general, for a set $V \subseteq \variables$ of variables, a \emph{trace} over $V$ is an infinite sequence $x_0x_1x_2\ldots \in (2^{\variables})^\omega$ of valuations of $V$. A \emph{specification} $\varphi \subseteq (2^{\variables})^\omega$ is a set of traces over $\variables$.
Two traces of disjoint sets $V,V'\subset\variables$ can be \emph{combined} by forming the union of their valuations at each position, i.e., 
$x_0x_1x_2\ldots \sqcup y_0y_1y_2\ldots = (x_0 \cup y_0)(x_1\cup y_1)(x_2\cup y_2) \ldots$. Likewise, the \emph{projection} of a trace onto
a set of variables $V' \subseteq \variables$ is formed by intersecting the valuations with $V'$ at each position:
$x_0x_1x_2\ldots \downarrow_{V'} = (x_0 \cap V')(x_1 \cap V')(x_2 \cap V')\ldots$.

For our specification language, we use propositional linear-time temporal logic (LTL)~\cite{LTL}, with the set $\variables$ of variables as atomic propositions and the usual temporal operators Next $\LTLnext$, Until $\LTLuntil$, Globally  $\LTLglobally$, and 
Eventually $\LTLeventually$.
System specifications are given as a conjunction $\varphi_p \wedge \varphi_q$ of two LTL formulas, where $\varphi_p$ refers only to variables in $O_p \cup O_e$, i.e., the formula relates the outputs of process $p$ to the outputs of the environment, and $\varphi_q$ refers only to variables in $O_q \cup O_e$. The two formulas represent the \emph{local specifications} for the two black-box processes. 
An implementation $s=(s_p, s_q)$ defines a set of traces 
\[\begin{array}{r}\mathit{Traces}(s_p, s_q) = \{ x_0x_1\ldots \in (2^{O})^\omega \mid 
x_k = s(i_0i_1\ldots i_{k-1}) \mbox{ for all } k \in \mathbb N \qquad\\ \mbox{ for some } i_oi_1i_2\ldots \in (2^{O_e})^\omega \}.\end{array}\] 
We say that the implementation \emph{satisfies} the specification if the traces of the implementation are contained in the specification, i.e., $\mathit{Traces}(s_p, s_q) \subseteq \varphi$.

\paragraph{The synthesis problem.}
Given an architechture and a specification~$\varphi$, the synthesis problem is to find an implementation~$s$  that satisfies~$\varphi$.
We say that a specification~$\varphi$ is \emph{realizable} in a given architecture if such an implementation exists, and \emph{unrealizable} if not.

\paragraph{Hyperproperties.}
We capture information-flow assumptions as hyperproperties. 
A \emph{hyperproperty over~$\variables$}  is a set $H \subseteq 2^{(2^{\variables})^\omega}$ of sets of traces over $\variables$~\cite{ClarksonS10}. An implementation $(s_p, s_q)$ satisfies the hyperproperty $H$ iff its traces are an element of $H$, i.e., $\mathit{Traces}(s_p, s_q) \in H$. 
A specification language for hyperproperties is the temporal logic HyperLTL~\cite{HyperLTL}. HyperLTL extends LTL with quantification over trace variables. The syntax of HyperLTL is given by the following grammar $
\varphi \Coloneqq \forall \pi \ldot \varphi \mid \exists \pi \ldot \varphi \mid \psi$ and $
\psi \Coloneqq  v_\pi 
\mid \neg \psi \mid \psi \land \psi \mid \LTLnext \psi \mid \psi \LTLuntil \psi
$
where $v_\pi\in \variables$ is a variable 
and $\pi \in \mathcal T$ is a trace variable. Note that the output variables are indexed by trace variables.
The quantification over traces makes it possible to express properties like ``\textit{$\psi$ must hold on all traces}'', which is expressed by $\forall \pi.~\psi\,$. Dually, one can express that ``\textit{there exists a trace on which $\psi$ holds}'', denoted by $\exists \pi.~\psi\,$. 
The temporal operators are defined as in LTL. 

In some cases, a hyperproperty can be expressed in terms of a binary relation on traces. 
A relation $R \subseteq (2^{\variables})^\omega \times (2^{\variables})^\omega$ of pairs of traces defines the hyperproperty $H$, where a set $T$ of traces is an element of $H$ iff for all pairs $\pi, \pi' \in T$ of traces in $T$ it holds that $(\pi, \pi') \in R$.
We call a hyperproperty defined in this way a  
\emph{2-hyperproperty}.
In HyperLTL, 2-hyperproperties are expressed as formulas with two universal quantifiers and no existential quantifiers.
A 2-hyperproperty can equivalently be represented as a set of infinite sequences over the product alphabet $\Sigma^2$: for a given 2-hyperproperty $R \subseteq \Sigma^\omega \times \Sigma^\omega$, let $R' = \{ (\sigma_0, \sigma_0') (\sigma_1,\sigma_1') \ldots \mid (\sigma_0\sigma_1\ldots, \sigma_0'\sigma_1'\ldots) \in R \}$. This representation is convenient for the use of automata to recognize 2-hyperproperties. 
\paragraph{Automata.}A \emph{nondeterministic finite automaton} (NFA) is a tuple $\A = (Q,\Sigma, q_0,\\ F, \delta)$, where $Q$ denotes a finite set of states, $\Sigma$ is a finite alphabet, $q_0$ is a designated initial state, $F\subseteq Q$ is the set of accepting states, and $\delta: Q \times \Sigma \rightarrow \mathcal{P}(Q)$ is the transition relation that maps a state and a letter to a set of possible successor states. 
A run of $\A$ on a finite word $w = w_0 \dots w_n \in \Sigma^*$ is a sequence of states $ r = q_0\dots q_{n+1} \in Q^*$ with $q_{i+1} \in \delta(q_i,w_i)$ for all $0 \le i \le n$. The run $r$ is accepting if $q_{n+1} \in F$. The set of all accepted words by an automaton $\A$ is called its \textit{language}, denoted by $\Lang(\A)$.
A \emph{B\"uchi automaton} $\mathcal B = (Q,\Sigma, q_0, F, \Delta)$ is an automaton over infinite words. A run of~$\mathcal B$ on an infinite word $w = w_1w_2 \dots \in \Sigma^{\omega}$  is an infinite sequence $ r = q_0q_1\dots \in Q^{\omega}$ with $q_{i+1} \in \delta(q_i,w_i)$ for all $i \in \mathbb N$. A run $r$ is accepting if there exist infinitely many $i \in \mathbb{N}$ such that $q_i \in F$. We use a B\"uchi automaton $\A$ over the alphabet $\Sigma^2$ to represent the 2-hyperproperty $R' \subseteq (\Sigma^2)^\omega$ with $\Lang(\A) = R'$.

\section{Necessary Information Flow in Distributed Systems}\label{sec:necessaryinformationflowindistributedsystems}
In reactive synthesis it is natural that the synthesized process reacts to different environment outputs.
This is also the case for distributed synthesis, where some outputs of the environment are not observable by a local process and the hidden values must be communicated to the process.
In the following we show when such information flow is necessary.
\subsection{Necessary Information Flow}\label{sec:Necessary:Information:Flow:Assumptions}
Our analysis focuses on pairs of situations for which the specification dictates a \emph{different} reaction from a given black-box process~$p$. Such pairs imply the need for
information flow that will enable $p$ to distinguish the two situations: if~$p$ cannot distinguish the two situations, it will behave in the same manner in both. Consequently, the specification will be violated, no matter how~$p$ is implemented, in at least one of the two situations.
A process $p$ needs to satisfy a local specification $\varphi_p$, which relates its outputs~$O_p$ to the outputs $O_e$ of the environment. (Recall that~$O_e$ may contain inputs to the other black-box process.) We are therefore interested in pairs of traces over $O_e$ for which~$\varphi_p$ does \emph{not} admit a common valuation of $O_p$. We collect such pairs of traces in a \emph{distinguishability relation}, denoted by $\Delta_p$:

\begin{definition}[Distinguishability]\label{def:distinguishability:relation}
Given a local specification~$\varphi_p$ for process~$p$, the \emph{distinguishability relation} $\Delta_p$ is the set of pairs of traces over $O_e$ (environment outputs) such that no trace over~$O_p$ satisfies $\varphi_p$ in combination with both traces in the pair. Formally: 
\begin{align*}
		\Delta_{p} = \{ (\pi_e, \pi_e') \in (2^{O_e})^\omega \times &(2^{O_e})^\omega \mid \ \\ 
		 \forall \pi_p \in\ &(2^{O_p})^\omega\ldot~ \mbox{if~ } \pi_e \sqcup \pi_p  \vDash \varphi_p \mbox{ ~then~ } \pi_e' \sqcup \pi_p \nvDash \varphi_p\ \}
\end{align*}
                
\end{definition}

By definition of $\Delta_p$, process~$p$ must distinguish 
$\pi_e$ from~$\pi'_e$, 
because it cannot respond to both in the same manner. 
In our running example, $\Delta_b$ consists of all pairs of sequences of values of $\vin$ that differ in
the first value of $\vin$. Process $b$ must act differently in such situations: if  $\vin$ is initially  $\mathit{true}$ then~$b$ must eventually set $\vout$ to $\mathit{true}$, while if it starts as~$\mathit{false}$, then~$b$ must keep $\vout$ always set to $\mathit{false}$.

In general, a black-box process $p$ must satisfy its specification $\varphi_p$ despite having only partial  access to~$O_e$. The distinguishability relation therefore directly defines an \emph{information flow} requirement: In order to satisfy $\varphi_p$, enough information about $O_e$ must be communicated to~$p$ via its local inputs~$I_p$ to ensure that~$p$ can distinguish any pair of traces in $\Delta_p$.
We formalize this information flow assumption as a 2-hyperproperty, which states that if the outputs of the environment in the two traces must be distinguished, i.e, the projection on $O_e$ is in $\Delta_p$, then there must be a difference in the local inputs $I_p$:

\begin{definition}[Information flow assumption]\label{def:necessary:information:flow:guarantee}
The \emph{information flow assumption} $\psi_p$ induced by~$\Delta_p$  
is the 2-hyperproperty defined by the relation 
\[
 R = \{ (\pi, \pi') \in (2^{\variables})^\omega \times (2^{\variables})^\omega \mid
 (\pi {\downarrow_{O_e}}, \pi' {\downarrow_{O_e}}) \in \Delta_p \mbox{ then } \pi {\downarrow_{I_p}} \neq \pi' {\downarrow_{I_p}} \}
\]
\end{definition}

In our running example, the information flow assumption for process $b$ requires that on any two executions that disagree on  the initial value of $\vin$,  the values
communicated to~$b$ over the channel $\vc$ must differ at some point. 
Observe that the  information flow assumption~$\psi_p$ specifies neither how the information is to be encoded on~$\vc$ nor the point in time when the different communication occurs. 
However, $\psi_p$ requires that the communication differs eventually if the initial values of $\vin$ are different. 
Moreover, notice that both~$\Delta_p$ and $\psi_p$ are determined by~$p$'s specification~$\varphi_p$.
The following theorem shows that the information flow assumption $\psi_p$ is a necessary condition.

\begin{theorem}\label{th:necessary:information:flow:assumption}
Every implementation that satisfies the local specification~$\varphi_p$  for~$p$ also satisfies the information flow assumption~$\psi_p$. 
\end{theorem}
\begin{proof}
Assume that there exists an implementation $(s_p, s_q)$ that satisfies $\varphi_p$ but not $\psi_p$. We show that this leads to a contradiction.
Since $\psi_p$ is not satisfied, there exists a pair of traces $\pi, \pi'$ such that $(\pi {\downarrow_{O_e}}, \pi' {\downarrow_{O_e}}) \in  \Delta_p$ and
$\pi {\downarrow_{I_p}} = \pi' {\downarrow_{I_p}}$. 
Let $\pi_e= \pi {\downarrow_{O_e}}$, and $\pi_e'= \pi' {\downarrow_{O_e}}$.
Since the inputs to process $p$ are the same on $\pi$ and $\pi'$, and since the strategies $s_p$ and $s_q$ are deterministic, the sequence of outputs is also the same.
Let $x_0x_1x_2\ldots = \pi {\downarrow_{I_p}} = \pi' {\downarrow_{I_p}}$ be the sequence of inputs.
We construct the sequence of outputs $o_0o_1o_2 \ldots$ generated by the implementation as follows: $o_k = s_p(x_0x_1 \ldots x_{k-1})$ for all $k\in \mathbb N$. Given that the implementation satisfies $\varphi_p$, we have that 
both $\pi_e \sqcup o$ and $\pi_e' \sqcup o$ satisfy $\varphi_p$. This, however, contradicts the assumption that $(\pi {\downarrow_{O_e}}, \pi' {\downarrow_{O_e}}) \in  \Delta_p$.
\end{proof}

\subsection{Time-bounded Information Flow}
\label{sec:Time:Bounded:Information:Flow}
We now introduce a strengthened version of the information flow assumption. As shown in Theorem~\ref{th:necessary:information:flow:assumption}, the information flow assumption is a necessary condition for the existence of an implementation that satisfies the specification. Often, however, the information flow assumption is not strong enough to allow for the separate synthesis of individual components in a compositional approach. 

Consider again process $b$ in our motivating example. The information flow assumption guarantees that any pair of traces that differ in the initial value of the global input $\vin$ will differ at some point in the value of the channel $\vc$. This assumption is not strong enough to allow process $b$ to satisfy the specification that $b$ must eventually set $\vout$ to $\mathit{true}$ iff the initial value of $\vin$ is $\mathit{true}$. Suppose that $\vin$ is $\mathit{true}$ initially. Then $b$ must at some point set $\vout$ to $\mathit{true}$. Process $b$ can only do so when it \emph{knows} that the initial value of $\vin$ is $\mathit{true}$. The information flow assumption is, however, too weak to guarantee that process $b$ will eventually obtain this knowledge. 
To see this, consider a hypothetical behavior of process $a$ that sets $\vc$ forever to $\mathit{true}$, if $\vin$ is $\mathit{true}$ in the first position, and if~$\vin$ is $\mathit{true}$ then $a$ keeps~$\vc$ true for~$n-1$ steps,  where $n>0$ is some fixed natural number, before it sets $\vin$ to $\mathit{false}$ at the $n^{\mathrm{th}}$ step. This behavior of process $a$ satisfies the information flow assumption for any number $n$; however, without knowing $n$, process $b$ does not know how many steps it should wait for $\vin$ to become $\mathit{false}$.
If, at any point in time $t$, the channel~$\vc$ has not yet been set to $\mathit{false}$, process $b$ can never rule out the possibility that the initial value of $\vin$ is $\mathit{true}$; it might simply be the case that $t<n$ and, hence, the time when $\vc$ will be set to $\mathit{false}$ still lies in the future of $t$! Hence, process $b$ can never actually set $\vout$ to $\mathit{true}$.

We begin by presenting a finer version of the distinguishability relation from Definition~\ref{def:distinguishability:relation} that we call \emph{time-bounded distinguishability}. Recall that by Definition~\ref{def:distinguishability:relation}, a pair  $(\pi_e, \pi_e')$ is in the  distinguishability relation~$\Delta_p$  if every output sequence~$\pi_p$ for~$p$ violates $p$'s specification~$\varphi_p$  when combined with at least one of the input sequences~$\pi_e$ or~$\pi_e'$.
Equivalently, if~$\varphi_p$ is satisfied by~$\pi_p$ combined with~$\pi_e$, then it is violated when~$\pi_p$ is combined with~$\pi_e'$. 
Observe that for~$p$ to behave differently in two scenarios, a difference must occur at a finite time~$t$. Clearly, this will only happen if $p$'s input shows a difference in finite time. To capture this, we say that a pair $(\pi_e,\pi_e')$ of environment output sequences is in the \emph{time-bounded} distinguishability relation if the violation with~$\pi_e'$ is guaranteed to happen in finite time. In order to avoid this violation, process~$p$ must act in finite time, before the violation occurs on~$\pi_e'$. 
We say that a trace $\pi$ \emph{finitely violates} an LTL formula $\varphi$, denoted by $\pi \nvDash_f \varphi$, if there exists a finite prefix $w$ of $\pi$ such that every (infinite) trace extending~$w$ violates $\varphi$.

\begin{definition}[Time-bounded distinguishability]\label{def:distinguishability:relation:2}
Given a local specification~$\varphi_p$ for process~$p$, the \emph{time-bounded distinguishability relation} $\Lambda_p$ is the set of pairs $(\pi_e,\pi_e')\in (2^{O_e})^\omega\times (2^{O_e})^\omega$ of traces of 
global inputs such that 
every trace of local outputs $\pi_p\in O_p$  either violates the specification $\varphi_p$ when combined with~$\pi_e$, or 
finitely violates~$p$'s local specification $\varphi_p$ when combined with $\pi_e'$:
\begin{align*}
		\Lambda_{p} = \{ (\pi_e, \pi_e') &\in (2^{O_e})^\omega \times (2^{O_e})^\omega \mid\\ 
		\forall \pi_p &\in (2^{O_p})^\omega\ldot~ \mbox{if~ } \pi_e \sqcup \pi_p\,  \vDash \varphi_p \mbox{ ~then~ } \pi_e' \sqcup \pi_p\, \nvDash_f \varphi_p\ \}
\end{align*}
\end{definition}
Note that, unlike the distinguishability relation $\Delta_p$, the \emph{time-bounded} distinguishability relation $\Lambda_p$ is not symmetric: For $(\pi_e, \pi_e')$, the trace $\pi_e' \sqcup \pi_p$ has to finitely violate~$\varphi_p$, while the trace $\pi_e \sqcup \pi_p$ only needs to violate~$\varphi_p$ in the infinite evaluation.
As a result, the corresponding \emph{time-bounded} information flow assumption will also be asymmetric: we require that on input $\pi_e$, process $p$ eventually obtains the knowledge that the input is different from~$\pi_e'$. For input $\pi_e'$ we do not pose such a requirement. 
The intuition behind this definition is that on environment output $\pi'_e$,  process~$p$ must definitely produce some output that does \emph{not} finitely violate $\varphi_p$. This output can safely be produced without ever knowing that the input is $\pi'_e$. However, on input $\pi_e$, it becomes necessary for process~$p$ to eventually deviate from the output that would work for $\pi_e'$. In order to safely do so, $p$ needs to realize after some finite time that the input is not $\pi_e'$. In our running example, $\pi_e$ would be an input in which $\vin$ is initially $\mathit{true}$, while $\pi_e'$ will be one in which it starts out being $\mathit{false}$.

Suppose we have a function $t: (2^{O_e})^\omega \rightarrow \mathbb N$ that identifies, for each environment output $\pi_e$, the time $t(\pi_e)$ by which process $p$ is guaranteed to know that the environment output is not $\pi_e'$. We define the information flow assumption for this particular function $t$ as a 2-hyperproperty. Since we do not know $t$ in advance, the time-bounded information flow assumption is the (infinite) union of all 2-hyperproperties corresponding to the different possible functions $t$.

\begin{definition}[Time-bounded information flow assumption]\label{def:time:bounded:information:flow:assumption}
Given the time-bounded distinguishability relation~$\Lambda_p$ for process~$p$, the \emph{time-bounded information flow assumption} $\chi_p$ for~$p$ 
is the (infinite) union over the 2-hyperproperties induced by the following relations~$R_t$,
for all possible functions $t: (2^{O_e})^\omega \rightarrow \mathbb N$: 
\begin{align*}
R_t = \{ (\pi&, \pi') \in (2^{\variables})^\omega \times (2^{\variables})^\omega \mid\\ 
& \mbox{ if } (\pi {\downarrow_{O_e}}, \pi' {\downarrow_{O_e}}) \in \Lambda_p, \mbox{ then } \pi[0..t(\pi{\downarrow_{O_e}})]{\downarrow_{I_p}} \neq \pi'[0..t(\pi{\downarrow_{O_e}})] {\downarrow_{I_p}} \}
\end{align*}
\end{definition}

Unlike the information flow assumption (cf.  Theorem~\ref{th:necessary:information:flow:assumption}), the \emph{time-bounded} information flow assumption is not in general a necessary assumption. 
Consider a modification of our motivating example, where there is an
additional environment output {\tt start}, which is only visible to process $a$, not to process $b$. The previous specification $\varphi_b$ is modified so that if {\tt in} is $\mathit{true}$ initially, then {\tt out} must be $\mathit{true}$ two steps after {\tt start} becomes $\mathit{true}$ for the first time; if {\tt in} is $\mathit{false}$ initially, then {\tt out} must become $\mathit{false}$ after two positions have passed since the first time {\tt start} has become $\mathit{true}$. The specification $\varphi_a$ ensures that the channel {\tt c} is set to $\mathit{true}$ until {\tt start} becomes $\mathit{true}$.
Clearly, this is realizable: if {\tt in} is $\mathit{false}$ initially, process $a$ sets {\tt c} to $\mathit{false}$ once {\tt start} becomes $\mathit{true}$, otherwise {\tt c} stays $\mathit{true}$ forever. Process $b$ starts by setting {\tt out} to $\mathit{true}$. It then waits for {\tt c} to become $\mathit{false}$, and, if and when that happens, sets {\tt out} to $\mathit{false}$. In this way, process $b$ accomplishes the correct reaction within two steps after {\tt start} has occurred. However, the function $t$ required by the time-bounded information flow assumption does not exist, because the time of the communication depends on the environment: the prefix needed to  distinguish an environment output $\pi_e$, where {\tt in} is $\mathit{true}$ initially from an environment output $\pi_e'$, where {\tt in} is $\mathit{false}$ initially, depends on the time when {\tt start} becomes $\mathit{true}$ on $\pi_e'$.

We now characterize a set of situations in which the time-bounded information flow requirement is still a necessary requirement. For this purpose we consider time-bounded distinguishability relations where the safety violation occurs after a bounded number of steps. We call such 
time-bounded distinguishability relations \emph{uniform}; the formal definition follows below.

\begin{definition}[Uniform distinguishability]
A time-bounded distinguishability relation $\Lambda_p$ is \emph{uniform}
if for every trace $\pi_e \in (2^{O_e})^\omega$ of global inputs, and every trace $\pi_p \in (2^{O_p})^\omega$ of local outputs of~$p$, there exists a natural number $n \in N$ such that for all $\pi_e' \in (2^{O_e})^\omega$ s.t.  $(\pi_e,\pi_e') \in \Lambda_p 
\mbox{~if~} \pi_e \sqcup \pi_p  \vDash \varphi_p \mbox{~then~} \pi_e' \sqcup \pi_p \nvDash_n \varphi_p$.
\end{definition}
\begin{theorem}\label{th:necessary:knowledge:assumption}
Let $\Lambda_p$ be a uniform time-bounded distinguishability relation derived from process~$p$'s local specification $\varphi_p$. Every computation tree that satisfies $\varphi_p$ also satisfies the time-bounded information flow assumption $\chi_p$. 
\end{theorem}
\begin{proof}
Let $(s_a, s_b)$ be an implementation that satisfies $\varphi_p$. We show that the time-bounded information-flow assumption $\chi_p$ is satisfied by defining a function $t: (2^{O_e})^\omega \rightarrow \mathbb N$ such that the 2-hyperproperty given by $R_t$ is satisfied. To compute $t(\pi_e')$ for some trace of inputs $\pi_e' \in (2^{O_e})^\omega$, we consider the trace of outputs $\pi_p' \in (2^{O_p})^\omega$ obtained by applying the implementation to the prefixes of $\pi_e'$. Since $\Lambda_p$ is uniform, there is a natural number $n \in \mathbb N$ such that for all $\pi_e$ with $(\pi_e,\pi_e') \in \Lambda_p$, we have that $\pi_e' \sqcup \pi_p \nvDash_n \varphi_p$. We set $t(\pi_e)$ to $n$. 

To convince yourself that $\chi_p$ is satisfied, suppose, by way of contradiction, that $R_t$ is violated on some pair $(\pi_e,\pi_e') \in \Lambda_p$ of input traces, i.e., the projection on $I_p$ is the same for $\pi_e$ and $\pi_e'$ on the entire prefix of length $t(\pi_e')$. But then, also the output of process $p$ must be the same along the entire prefix; this, however, means that $\pi_e'$ will violate $\varphi_p$ after $n=t(\pi_e')$ steps, contradicting our assumption that the implementation satisfies $\varphi_p$.
\end{proof}

\section{Computing Information Flow Assumptions}\label{app:computinginformationflowassumptions}
\subsection{Automata for Information Flow Assumptions}
We first give an explicit construction of an automaton that recognizes the information flow assumption~$\psi_p$ that is induced by~$\varphi_p$.
The local specification $\varphi_p$ is given as an LTL formula, which can be translated into an 
equivalent B\"uchi automaton $\A_{\varphi}$ over alphabet $2^{O_e \cup O_p}$ \cite{ReasoningAboutInfiniteComputationPaths}.
We self-compose $\A_{\varphi}$ into an automaton ${\mathcal B}$ over the alphabet $2^{O_e \cup O_p} \times 2^{O_e \cup O_p}$ such that  ${\mathcal B}$ accepts a sequence of pairs iff both the projection on the first components and the projection on the second components are accepted by $\A_\varphi$ and, additionally, both components always agree on the values of $O_p$. We then construct a B\"uchi automaton $\mathcal C$ over the alphabet $2^{O_e} \times 2^{O_e}$ that guesses the values of $O_p$ nondeterministically so that a pair of sequences is accepted by $\mathcal C$ iff there exists a valuation of $O_p$ such that the extended sequences are accepted by $\mathcal B$.
The automaton $\mathcal C$ thus accepts all sequences of global inputs that process $p$ does \emph{not} need to distinguish, because there is a sequence of outputs that satisfies the specification in both cases. 
We construct another B\"uchi automaton $\mathcal D$ over alphabet $2^{I_p} \times 2^{I_p}$ that recognizes a sequence of pairs of local input values iff they differ at some point. Finally, we construct a B\"uchi automaton~$\mathcal E$ over the alphabet $2^{O_e \cup I_p} \times 2^{O_e \cup I_p}$ that accepts a sequence of pairs iff the sequence of projections on $O_e$ is accepted by $\mathcal C$ or the sequence of projections on $I_p$ is accepted by $\mathcal D$.
The automaton $\mathcal E$ recognizes the information flow assumption~$\psi_p$ of process~$p$.

\begin{theorem} \label{thm:automatonforifa}
For a process $p$ with local specification $\varphi_p$, there exists a B\"uchi automaton with an exponential number of states in the length of $\varphi_p$ that recognizes the information flow assumption~$\psi_p$ induced by~$\varphi_p$.
\end{theorem}
\begin{proof}
The automaton $\mathcal E$ described above recognizes~$\psi_p$. We now claim that it has the stated size. The number of states of $\A_{\varphi}$ is exponential in the length of $\varphi_p$. By construction, the number of states of $\mathcal B$ is quadratic in the number of states of $A_{\varphi}$, and  $\mathcal C$ has the same number of states as $\mathcal B$. The automaton~$\mathcal D$ needs only two states. Hence, $\mathcal E$ has only two more states than $\mathcal B$ and so its total number  of states is  exponential in the length of $\varphi_p$, as claimed. \end{proof}

\subsection{Checking Uniformity}

We begin with the construction of an automaton $\mathcal A_{\Lambda_p}$ over alphabet $2^{O_e} \times 2^{O_e}$ that recognizes the time-bounded distinguishability relation $\Lambda_p$. 
Let $\mathcal A_{\neg \varphi_p}$ be a deterministic $\omega$-automaton over alphabet $2^{O_e \cup O_p}$ that recognizes all traces that violate the local specification $\varphi_p$. Let $\mathcal B_{\neg \varphi_p}$ be a deterministic finite-word automaton over alphabet $2^{O_e \cup O_p}$ that recognizes the bad prefixes of $\varphi_p$. We combine $\mathcal A_{\neg \varphi_p}$ and $\mathcal B_{\neg \varphi_p}$ into a deterministic $\omega$-automaton $\mathcal C$ over alphabet $2^{O_e} \times 2^{O_e} \times 2^{O_p}$ that accepts traces of two inputs $\pi_e,\pi_e'$ and an output $\pi_p$ such that $\pi_e \sqcup \pi_p$ violates $\varphi_p$ or $\pi_e' \sqcup \pi_e$ \emph{finitely} violates $\varphi_p$. We obtain the universal automaton $\mathcal D_{\Lambda_p}$ with alphabet $2^{O_e} \times 2^{O_e}$ as the universal projection of $\mathcal C$ with respect to the outputs $\pi_p$.

\begin{theorem}\label{thm:timeboundeddistinguishability}
    For  a  process $p$ with  local  specification $\varphi_p$, there exists a universal $\omega$-automaton $\mathcal A_{\Lambda_p}$ over alphabet $2^{O_e} \times 2^{O_e}$ that recognizes the time-bounded distinguishability relation $\Lambda_p$. The number of states of $\mathcal A_{\Lambda_p}$ is doubly-exponential in the length of $\varphi_p$.
\end{theorem}
\begin{proof}
Both $\mathcal A_{\neg \varphi_p}$ and $\mathcal B_{\neg \varphi_p}$ have doubly-exponentially many states in the length of $\varphi_p$~\cite{Kupferman+Vardi/99/Safety}. The size of $\mathcal C$ is the product of the sizes of $\mathcal A_{\neg \varphi_p}$ and $\mathcal B_{\neg \varphi_p}$. Because of the universal projection, $\mathcal D_{\Lambda_p}$ is universal, rather than deterministic, but still of doubly-exponential size.
\end{proof}

Next, we check whether the time-bounded distinguishability relation is uniform. We construct an automaton that recognizes all traces of inputs and local outputs where no 
uniform bound exists. Let $\mathcal A_{\varphi_p}$ be a universal $\omega$-automaton over alphabet $2^{O_e \cup O_p}$ that recognizes all traces that satisfy the local specification $\varphi_p$. We combine $\mathcal A_\varphi$ with 
$\mathcal D_{\Lambda_p}$ to a universal $\omega$-automaton $\mathcal E$ over alphabet  $2^{O_e} \times 2^{O_e} \times 2^{O_p}$ that accepts traces of two inputs $\pi_e,\pi_e'$ and an output $\pi_p$ when $(\pi_e,\pi_e') \in \Lambda_p$ and $\pi_e \sqcup \pi_p \models \varphi_p$. 
From $\mathcal E$ we construct a universal automaton $\mathcal F$ over alphabet $2^{f} \times 2^{O_e} \times 2^{O_p}$ that accepts $\pi_e$ and $\pi_p$ if there exists an $\pi_e'$ such that the bad prefix is reached on $\pi_e'$ after $f$ becomes $\mathit{true}$ for the first time.
Finally, we obtain a universal automaton $\mathcal G$ over alphabet $2^{O_e} \times 2^{O_p}$ that accepts those $\pi_e$ and $\pi_p$ that are accepted by $\mathcal F$ for all traces of $f$ that set $f$ to $\mathit{true}$ at least once.
\begin{theorem}
    For  a  process $p$ with  local  specification $\varphi_p$, whether
    the time-bounded distinguishability relation is uniform can be checked in quadruply exponential running time.
\end{theorem}
\begin{proof}
$\mathcal A_{\varphi_p}$ is exponential in the length of $\varphi_p$,
$\mathcal D_{\Lambda_p}$ is doubly exponential; hence, $\mathcal E$ is also doubly exponential. Because of the projection in $O_e$, $\mathcal F$ is triply exponential. Because $\mathcal F$ is universal, the universal projection in $f$ does not cause a further increase in the number of states, the size of $\mathcal G$ is thus triply exponential. Emptiness of a universal automaton can be checked in exponential time, resulting in an overall quadruply exponential running time.
\end{proof}

\subsection{Computing Time-bounded Information Flow Assumptions}

Our final goal in this section is to compute \emph{time-bounded} information flow assumptions. Time boundedness introduces a new difficulty, because an unbounded number of traces is required to satisfy the same bound; hence, the time-bounded information flow assumption is not a $k$-hyperproperty for any value $k \in \mathbb N$. In the following, we nonetheless represent the time-bounded information flow property as a 2-hyperproperty, by employing the following trick: We introduce a fresh atomic proposition $t$, which is to be read by process $p$ as a new input and is to be computed by process $p$'s environment. The first occurrence of $t$ indicates that the time bound has been reached. This extra proposition allows us to express the time-bounded information flow assumption as a 2-hyperproperty: we first require that $t$ occurs on every trace  that appears as a left trace in $\Lambda_p$  (condition~1). Furthermore,  process $p$ must observe a difference between any pair of traces in $\Lambda_p$ before $t$ occurs on the left trace (condition~2).

We begin with the universal automaton $\mathcal A_{\Lambda_p}$ over alphabet $2^{O_e} \times 2^{O_e}$ from Theorem~\ref{thm:timeboundeddistinguishability}, which recognizes the time-bounded distinguishability relation $\Lambda_p$. We dualize $\mathcal A_{\Lambda_p}$ to obtain the nondeterministic automaton 
$\overline{\mathcal A_{\Lambda_p}}$ that recognizes all pairs of traces \emph{not} in $\Lambda_p$. 
For condition 1, we construct a nondeterministic automaton $\mathcal H_1$ that checks that $t$ occurs on the left trace; for condition 2, we construct a nondeterministic automaton $\mathcal H_2$ that ensures that the traces differ in the local inputs before $t$ occurs. Combining $\overline{\mathcal A_{\Lambda_p}}$ with $\mathcal H_1$ and $\mathcal H_2$, we obtain a nondeterministic automaton $\mathcal I$ over the alphabet $2^{O_e \cup I_p \cup O_p \cup \{t\}} \times 2^{O_e \cup I_p \cup O_p \cup \{t\}}$ that represents the time-bounded information flow assumption.
\begin{theorem} 
For a process $p$ with local specification $\varphi_p$, there exists a nondeterministic $\omega$-automaton with a doubly-exponential number of states in the length of $\varphi_p$ that recognizes the time-bounded information flow assumption~$\chi_p$ induced by~$\varphi_p$.
\end{theorem}
\begin{proof}
The automaton $\mathcal I$ described above recognizes~$\chi_p$. We now claim that it has the stated size. By Theorem~\ref{thm:timeboundeddistinguishability}, the number of states of $\mathcal A_{\Lambda_p}$ is doubly-exponential in the length of $\varphi_p$. The dual $\overline{\mathcal A_{\Lambda_p}}$ has the same size as $\mathcal A_{\Lambda_p}$; finally, $\mathcal H_1$ and $\mathcal H_2$ each has a constant number of states. Thus, the number of states of $\mathcal I$ is also doubly-exponential in the length of $\varphi_p$.
\end{proof}

\section{Compositional Synthesis}\label{sec:behaviour}
We now use the time-bounded information flow assumptions to split the distributed synthesis problem for an architecture $(I_p,I_q,O_p,O_q,I_e)$ into two separate synthesis problems.
The local implementations are then composed and form a correct system, whose decomposition returns the solution for each process.

\subsection{Constructing the Hyper Implementations}
We begin with the synthesis of local processes.
Let $\Lambda_p$ and $\Lambda_q$ be the time-bounded distinguishability relations for $p$ and $q$, and let $\chi_p$ and $\chi_q$ be the resulting time-bounded information flow assumptions. In the individual synthesis problems, we ensure that process $p$ provides the information needed by process $q$, i.e., that the implementation of $p$ satisfies $\chi_q$, and, similarly, that $q$ provides the information needed by $p$, i.e.,  $q$'s implementation satisfies $\chi_p$. 

We carry out the individual synthesis of a process implementation on trees that branch according to the input of the process (including $\vt_p$) \emph{and} the environment's output. In such a tree, the synthesized process thus has access to full information. 
We call this tree a \emph{hyper implementation}, rather than an implementation, because the hyper implementation describes how the process will react to certain information, without specifying \emph{how} the process will receive information. 
This detail is left open until we know the other process' hyper implementation: at that point, both hyper implementations can be turned into standard strategies, which are trees that branch according to the process' own inputs. 

\begin{definition}[Hyper implementation]
   Let $p$ and $q$ be processes and $e$ be the environment.
    A $2^{O_e \cup I_p \cup \{\vt_p\}}$-branching $2^{O_p\cup \{\vt_q\}}$-labeled tree $h_p$ is a \emph{hyper implementation} of p.
\end{definition}
Since the hyper implementation has access to the full information, while the time-bounded information flow assumption only guarantees that the relevant information arrives after some bounded time, the strategy has ``too much'' information. We compensate for this by introducing a \emph{locality condition}: on two traces $(\pi_e, \pi_e') \in \Lambda_p$ in the distinguishability relation of process $p$, as long as the input to the process from the external environment is identical, process $p$'s output must be identical until $\vt_p$ happens (which signals that the bound for the transmission of the information has been reached). 
For traces $(\pi_e, \pi_e') \not\in \Lambda_p$ outside the distinguishability relation, process $p$'s output must be identical until there is a difference in the input to process $p$ or in the value of $\vt_p$.

\begin{definition} [Locality condition]\label{def:locality:condition}
Given the time-bounded distinguishability relation~$\Lambda_p$ for process~$p$, the \emph{locality condition} $\eta_p$ for~$p$ is the 2-hyperproperty induced by the following relation~$R$:
\begin{align*}
R = \{ (\pi&, \pi') \in (2^{O_e \cup I_p \cup \{t_p\}})^\omega \times (2^{O_e \cup I_p \cup \{t_p\}})^\omega \mid\\ 
& \mbox{ if } (\pi {\downarrow_{O_e}}, \pi' {\downarrow_{O_e}}) \in \Lambda_p,
\mbox{ then } \pi[0..t]{\downarrow_{O_p}} = \pi'[0..t] {\downarrow_{O_p}}  \mbox{ and}\\
& \mbox{ if } (\pi {\downarrow_{O_e}}, \pi' {\downarrow_{O_e}}) \not\in \Lambda_p,
 \mbox{ then } \pi[0..t']{\downarrow_{O_p}} = \pi'[0..t'] {\downarrow_{O_p}} \} 
\end{align*}
where $t$ is the smallest natural number such that $\vt_p \in \pi[0..t]$ or $\pi[0..t]\downarrow_{I_p} \neq \pi'[t]\downarrow_{I_p}$ (and $\infty$ if no such $t$ exists), and $t'$ is the smallest natural number such that $\pi[0..t']\downarrow_{I_p} \neq \pi'[0..t']\downarrow_{I_p}$ or $\pi[0..t']\downarrow_{\{\vt_p\}}\neq\pi'[0..t']\downarrow_{\{\vt_p\}}$ (and $\infty$ if no such $t'$ exists).
\end{definition}

We use HyperLTL to formulate the locality condition for process $b$ in our running example.
Based on the time-bounded distinguishability relation $\Lambda_b$, which relates every trace with $\vin$ in the first step to all traces on which $\neg \vin$ holds there, we can write the locality condition:
\begin{align*}
    \forall \pi,\pi'. (\vin_{\pi} &\wedge \neg \vin_{\pi'}) \rightarrow  \left((\vt_{\pi} \vee \vc_{\pi} \nleftrightarrow \vc_{\pi'} )\LTLrelease (\vout_{\pi} \leftrightarrow \vout_{\pi'})\right)\\
    \wedge (\neg(\vin_{\pi} &\wedge \neg\vin_{\pi'})) \rightarrow  \left(\vt_{\pi} \nleftrightarrow \vt_{\pi'} \vee \vc_{\pi} \nleftrightarrow \vc_{\pi'} )\LTLrelease (\vout_{\pi} \leftrightarrow \vout_{\pi'})\right) 
\end{align*}
The order in the formula is analogous to the order in \Cref{def:locality:condition}.
For all pairs of traces that are in the distinguishability relation, i.e., $\vin$ is $\mathit{true}$ on $\pi$ and $\mathit{false}$ on $\pi'$, the outputs being equivalent on both traces can only be released by $\vt$ on trace $\pi$ or by a difference in the local inputs ($\vc$).
Moreover, if the traces are not in the distinguishability relation, i.e., $\neg(\vin_{\pi} \wedge \neg \vin_{\pi'})$, then only a difference in $\vt$ or $\vc$ can release $\vout$ to be equivalent on both traces.
With the locality condition at hand, we define when a hyper implementation is locally correct:

\begin{definition}[Local correctness of hyper implementations]\label{def:local:correctness:hyper:implementations}
Let $p$ and $q$ be processes, let  $\varphi_p$ be the local specification of $p$, let $\eta_p$ be its locality condition, and let $\chi_q$ be the information flow assumption of $q$.
The hyper implementation $h_p$ of $p$ is \emph{locally correct} if it satisfies $\varphi_p$, $\eta_p$, and $\chi_q$.
\end{definition}
The specification $\varphi_p$ is a trace property, while $\eta_p$ and $\chi_q$ are hyperproperties. 
Since all properties that need to be satisfied by the process are guarantees, it is not necessary to assume explicit behaviour of process $q$ to realize process $p$.
Local correctness relies on the guarantee that the other process satisfies the current process' own information flow assumption.
Note that both the locality condition and the information flow assumption for $p$ build on the time-bounded distinguishability relation of $p$.

\subsection{Composition of Hyper Implementations}\label{sec:compositionalsynthesis}
The hyper implementations of each of the processes are locally correct and satisfy the information flow assumptions of the other process respectively.
However, the hyper implementations have full information of the inputs and are dependent on the additional variables $\vt_p$ and $\vt_q$.
To construct  practically executable  local implementations, we first compose the hyper implementations into one strategy.

\begin{definition} [Composition of hyper implementations]\label{def:composition:hyper:strategies}
Let $p$ and $q$ be two processes with hyper implementations given as infinite ${2^{O_e \cup I_p \cup \{\vt_p\}\cup \informedness_p}}$-branching $2^{O_p\cup \{\vt_q\}}$-labeled tree $h_p$ for process $p$, and an infinite $2^{O_e \cup I_q \cup \{\vt_q\}\cup \informedness_q}$-branching $2^{O_p\cup \{\vt_p\}}$-labeled tree $h_q$ for process $q$. 

Given two hyper implementations $h_p$ and $h_q$, we define the composition $h = h_p || h_q$ to be a $2^{O_e}$-branching $2^{O_p \cup O_q}$-labeled tree, where $h(v) = (h_p(f_p(v)) \cup h_q(f_q(v))) \cap (O_p \cup O_q)$ and $f_p, f_q$ are defined as follows:
\begin{align*}
     f_p(\epsilon) &= \epsilon & f_p(v \cdot x) &= f_p(v) \cdot ((x \cap I_p) \cup (h_q(f_q(v)) \cap (I_p \cup \{\vt_p\}))\\
    f_q(\epsilon) &= \epsilon  & f_q(v \cdot x) &= f_q(v) \cdot ((x \cap I_q) \cup (h_p(f_p(v)) \cap (I_q \cup \{\vt_q\}))
\end{align*}
\end{definition}

If each hyper implementation satisfies the time-bounded information flow assumption of the other process, then there exists a strategy for each process (given as a tree that branches according to the local inputs of the process), such that the combined behavior of the two strategies corresponds exactly to the composition of the hyper implementations.

The composition of the hyper implementations of the bit transmission protocol is shown in \Cref{fig:composition}.
The initial state is the combination of both processes initial states with the corresponding outputs.
We change the state after the value of $\vin$ is received. 
While process $a$ directly reacts to $\vin$, process $b$ cannot observe its value, and the composition can either be in $h^b_0$ or $h^b_1$.
Booth states have the same output.
In the next step, process $a$ communicates the value of $\vin$ by setting $\vc$ to $\text{true}$ or $\text{false}$, such that the loop states $h^a_1, h^a_1$ and $h^a_2, h^b_3$ are reached. 
\begin{figure}[t]
     \begin{center}
     \resizebox{.8\linewidth}{!}{
    \tikzstyle{state}=[draw, circle, fill=none, minimum width=1cm, 
minimum height = 1.3cm, align=center, thick]

\begin{tikzpicture}[->,>=stealth',shorten >= 1pt,auto]

\node[state] (p0)[label=above left:{$h_0^a, h^b_0$}] {%
    $\vc,$\\ 
    $\neg\vout$
};

\node (left) [left = 0.5 of p0]{};

\node[state] (p1) [above right = -0.1 and 1.5 of p0, label=above left:{$h_1^a, (h^b_1, h^b_2)$}] {%
  $\vc,$ \\
  $\neg \vout$
 };
\node[state] (p2) [below right= -0.1 and 1.5 of p0, label=below left:{$h_2^a, (h^b_1, h^b_2)$}]{
    $\neg \vc,$\\
    $\neg \vout$
    };

\node[state] (p3) [right = 1.2 of p1, label=above right:{$h_1^a, h^b_1$}] {%
    $\vc,$\\
    $\neg \vout$
  };
  
\node[state] (p4) [right = 1.2 of p2, label=below right:{$h_2^a, h^b_3$}] {%
    $\neg \vc$\\
    $\vout$
  };

\node (format) [right = 2.2 of p4]{};

\path (left) edge (p0)
      (p0) edge[thick] node[above left, xshift=3pt,yshift=1pt,] {%
        $\neg \vin$
      } (p1)
      (p0) edge[thick] node [below left, xshift=7pt,yshift=-1pt, align=center] {%
      $ \vin$
      }(p2)
       (p0) edge[loop left, thick, draw=none] node [left] {%
       $\phantom{\LTLtrue}$}
       (p0) 
     (p2) edge[thick] node[below] {
      $\ast$
      }(p4)
      (p1) edge[thick] node[above] {
      $\ast$
      }(p3)
      (p4) edge[loop right, thick] node [right] {%
      $\LTLtrue$}(p4) 
      (p3) edge[loop right, thick] node [right] {%
      $\LTLtrue$}(p3)
      ;
\end{tikzpicture}
    }
    \end{center}
    \caption{The composition of the hyper implementations of $a$ in \Cref{fig:implemetation of -a} and $b$ in \Cref{fig:implementation-b}. The states are labeled with the combination of states reached for both processes, and multiple, if they cannot be distinguished.}
    \label{fig:composition}
\end{figure}
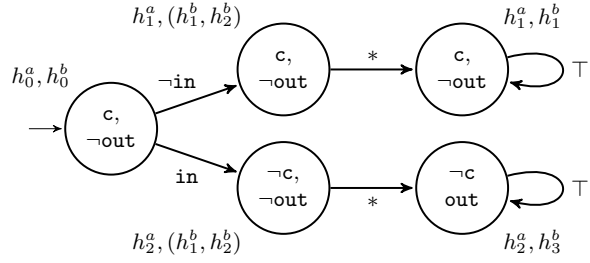

The local strategies of the processes are constructed from the composed hyper implementations. As an auxiliary notion we introduce the \emph{knowledge set}: the set of finite traces in the composition that cannot be distinguished by a process.

\begin{definition}[Knowledge set]\label{def:knowledge:set}
Let $p$ and $q$ be two processes with composed hyper implementations $h = h_p || h_q$.
For a finite trace $v \in (2^{I_p})^*$ of inputs to $p$, we define the knowledge set $K_p(v)$ 
to be
\[K_p(v)\triangleq\{w~|~\mbox{$w$ is a finite trace of $(2^{O_e})^{*}$ and }f_p(w) = v\}.\]
\end{definition}

\begin{lemma}\label{lem:knowledge:set}
For all s $v, v' \in {(2^{I_p})}^*$, if $K_p(v) = K_p(v')$ then $h(v)\downarrow_{O_p} = h(v')\downarrow_{O_p}$.
\end{lemma}
\begin{proof}
If $K_p(v)$ is a singleton or empty, then the lemma is trivially true.
Assume $|K_p(v)| \geq 1$ and there exists $w, w' \in K_p(v)$ s.t. $h(w)\downarrow_{O_p} \neq h(w')\downarrow_{O_p}$.
Since $w$ and $w'$ agree on the local inputs to $p$, there exists at least one $a \in O_e \backslash I_p$ s.t. $w \downarrow_{O_{a}} \neq w' \downarrow_{O_{a}}$. Then, $h_p(w) \neq h_p(w')$ has to hold 
following the function $f_p$ of \Cref{def:composition:hyper:strategies}.
Given the locality from \Cref{def:locality:condition}, this is only possible if $t_p$ was observed in the input to $h_p$, which is replaced by the output of $h_q$ in \Cref{def:composition:hyper:strategies}.
Since $h_q$ satisfies the time bounded information flow assumption $\chi_p$ from \Cref{def:time:bounded:information:flow:assumption}, $h_p$ observes a difference in $I_p$ before it reacts to the global inputs.
Therefore, $h(w)\downarrow_{O_p} = h(w')\downarrow_{O_p}$ which contradicts the assumption.
\end{proof}

The local strategies from the composed hyper implementations are then defined as follows:

\begin{definition}[Local strategies from hyper implementations]\label{def:local:strategies}
Let $p$ and $q$ be two processes with time-bounded information flow assumptions $\chi_p$ and $\chi_q$, and $h = h_p || h_q$ be the composition of their hyper implementations.
For $j \in \{p,q\}$ the strategy  $s_j$, represented as a $2^{I_j}$-branching $2^{O_j}$-labeled tree for process~$j$, is defined as follows: 
\[
    s_j(\epsilon) = \epsilon \phantom{space}
    s_j(v) = \twopartdef{\emptyset}{|K_j(v)| = 0}{h(min(K_j(v)))\downarrow_{O_j}}{|K_j(v)| > 0}
\]
where $min(K_j(v))$ is the smallest trace based on an arbitrary order over $K_j(v)$.
\end{definition}
The base case of the definition inserts a label for unreachable traces in the composed hyper implementation.
For example, the local inputs $I_p \backslash O_e$ are determined by $s_q$, 
and not all input words in $(2^{I_q})^*$ are possible.
Process $p$'s local strategy $s_p$  can discard these input words.
The second case of the definition picks the smallest trace in the knowledge set and computes the outputs from $h$ that are local to a process.
Intuitively, the outputs of $h$ have to be the same for every trace that a process considers possible in the composed hyper implementations. 
We therefore pick one of them, compute the output of the composed hyper-strategy, and restrict the output to the local outputs of the process.
The following theorem states the correctness of the construction in \Cref{def:local:strategies}.

\begin{theorem}\label{th:correctness:decomposition}
Let $p$ and $q$ be two processes with time-bounded information flow assumptions $\chi_p$ and $\chi_q$, let $h = h_p || h_q$ be the composition of the hyper implementations, and $s_p$ and $s_q$ be the local strategies.
Then, for all $v \in (2^{O_e})^*$ it holds that 
$h(v) = s_p(g_p(v)) \cup s_q(g_q(v))$ where
$g_p, g_q$ are defined as follows:
\begin{align*}
    g_p(\epsilon) &= \epsilon &g_p(v \cdot x) &= g_p(v) \cdot ((x \cap I_p) \cup (s_q(g_q(v)) \cap I_p)\\
    g_q(\epsilon) &= \epsilon &g_q(v \cdot x) &= g_q(v) \cdot ((x \cap I_q) \cup (s_p(g_p(v)) \cap I_q)
\end{align*}
\end{theorem}
\begin{proof}
Proof by induction over $v \in (2^{O_e})^*$.
\emph{Base case:} Let $v = \epsilon$, then $s_p(g_p(\epsilon)) \cup s_q(g_q(\epsilon)) = h(\epsilon) = \epsilon$. 
\emph{Induction Step:} The induction step is shown from  $v \in (2^{O_e})^*$ to $v \cdot x \in (2^{O_e})^*$, with $x \in 2^{O_e}$.
Inserting $g_p$ from \Cref{th:correctness:decomposition} we obtain $g_p(v \cdot x) = g_p(v) \cdot ((x \cap I_p) \cup (s_q(g_q(v) \cap I_p))$.
Since $s_q(g_q(v))$ and $g_p(v)$ is assumed correct, we show that the input trace returned by $g_p$ and given to $s_p$ is correct: The input is local to $p$ because $x \cap I_p$ and $s_q(g_q(v))\cap I_p)$ remove unobservable inputs, and all outputs of the previous step from $q$ are added to the current input.
It remains to show that the outputs of the local strategies combined are equal to the output of $h$: $h(v\cdot x) = s_p(g_p(v \cdot x)) \cup s_q(g_q(v\cdot x))$.
Let $v' \cdot x' = g_p(v\cdot x)$. 
Given \Cref{def:local:strategies} and \Cref{def:knowledge:set}, we know that $s_p(v' \cdot x') = h(K(w))\downarrow_{O_p}$, with $w \in K_p(v' \cdot x')$.
Since $K_p(v')$ is assumed correct, we show that adding $x'$ to $v'$ still results in correctness of $h(K(w))$.
Following \Cref{lem:knowledge:set}, all elements in $K(v' \cdot x')$ and therefore all corresponding paths in $h$ have the same label and picking any with $min(K_p(v' \cdot x'))$ is correct.
It follows that $h(v \cdot x)\downarrow_{O_p} = s_p(g_p(v \cdot x))$.
Using the same argument for $s_q$ by interchanging $p$ and $q$ in every index yields the correctness of the theorem, i.e., for all $v \in (2^{O_e})^*$ it holds that 
$h(v) = s_p(g_p(v)) \cup s_q(g_q(v))$.
\end{proof}

Combining all definitions and theorems of the previous sections, we conclude with the following corollary.
\begin{corollary}
Let $(I_p,I_q,O_p,O_q,I_e®)$ be an architecture and $\varphi = \varphi_p \wedge \varphi_q$ be a specification. 
If the hyper-strategies $h_p$ and $h_q$ are locally correct, then the implementation $(s_p, s_q)$ satisfies $\varphi$.
\end{corollary}

\section{A More Practical Approach}\label{sec:amorepracticalapproach}
\label{sec:practical}
A major disadvantage of the synthesis approach of the preceding sections is that the hyper implementations are based on the full set of environment outputs; as a result, hyper implementations branch according to inputs that are not actually available; this, in turn, necessitates the introduction of the locality condition. 

In this section, we develop a more practical approach, where the branching is limited to the information that is actually available to the process: this includes any environment output directly visible to the process and, additionally, the information the process is guaranteed to receive according to the information flow assumption. As a result, the synthesis of the process is sound without need for a locality condition.
We develop this approach under two assumptions: First, we assume that the time-bounded information flow assumption only depends on environment outputs the sending process can actually see; second, we assume that the time-bounded information flow assumption can be decomposed into a finite set of classes in the following sense: For a trace $\pi$ of environment outputs, the information class $[\pi]_p$ describes that, on the trace $\pi$, the process $p$ eventually needs to become aware that the current trace is in the set $[\pi]$. The information class is obtained by collecting all traces that are \emph{not} related to $\pi$ in the time-bounded distinguishability relation.

\begin{definition}[Information classes]
Given a time-bounded distinguishability relation~$\Lambda_p$ for process~$p$, the \emph{information class} $[\pi]_p$ of a trace $\pi$ over $O_e$ is the following set of traces:
$\quad
[\pi]_p = (2^{O_e})^\omega  \setminus \{ \pi' \in (2^{O_e})^\omega  \mid (\pi, \pi') \in \Lambda_p \}
$
\end{definition}

The next definition relativizes the specification of the processes for a particular information class, reflecting the fact that the process does not know the actual environment output, but only its information class; hence, the process output needs to be correct for all environment outputs in the information class.

\begin{definition}[Relativized specification]
For a process $p$ with specification $\varphi_p$ and an information class $c$, the relativized specification $\varphi_{p,c}$ is the following
trace property over $(I_p \cap O_e) \cup O_p$:
\begin{align*}
\varphi_{p,c} = \{ \pi_e \sqcup \pi_p \mid \pi_e \in (2^{I_p \cap O_e})^\omega, \pi_p \in (2^{O_p})^\omega \mbox{ s.t. } \forall \pi_e' \in c .\ \pi'_e \sqcup \pi_p \models \varphi_p \}
\end{align*}
\end{definition}

The component specification, which is the basis for the synthesis of the process, must take into account that the process does not know the information class in advance; the behavior of the other process will only eventually reveal the information class.
Let $IC$ be the set of information classes for process $p$. 
Assume that this set is finite. 
We now replace the inputs of the process that come from the other process with new input channels $IC$ as new inputs. 
In the hyper implementation, receiving such an input reveals the information class to the process. 
In the actual implementation, the information class will be revealed by the actual outputs of the other process that are observable for $p$.
The component specification requires that the processes satisfy the relativized specification under the assumption that the information class is eventually received. We encode this assumption as a trace condition $\psi$, which requires that exactly one of the elements of $IC$ eventually occurs.

\begin{definition}[Component specification]
\label{def:componentspec}
For process $p$ with specification $\varphi_p$, the component specification $\langle \varphi_{p}\rangle$ over $(I_p\cap O_e)\cup IC\cup O_p$ is defined as
\begin{align*}
\langle \varphi_{p} \rangle = \{ \pi \in (2^{(I_p \cap O_e) \cup IC \cup O_p})^\omega \mid \mbox{ if } \pi \models \psi \mbox{ then } \pi \models \bigwedge_{c \in IC} (\LTLdiamond c \rightarrow \varphi_{p,c} \}
\end{align*}
where $\psi$ is the following trace property over $(I_p \cap O_e) \cup IC \cup O_p$:
\begin{align*}
\psi = \{ \pi \in (2^{(I_p \cap O_e) \cup IC \cup O_p})^\omega \mid \exists \pi' \in (2^{O_e})^\omega.\ 
\pi\downarrow_{I_p \cap O_e} = \pi'\downarrow_{I_p \cap O_e}\\ \mbox{ and } \pi \models \LTLdiamond [\pi'] \mbox{ and exactly one element of } IC \mbox{ occurs on } \pi\} 
\end{align*}
\end{definition}
\begin{figure}[t]
     \centering
     \begin{subfigure}[b]{0.49\textwidth}
         \centering
         \resizebox{.9\linewidth}{!}{
         \tikzstyle{state}=[draw, rectangle, fill=none, minimum width=.8cm, 
minimum height = .8cm, rounded corners=1mm, align=center, thick]

\begin{tikzpicture}[->,>=stealth',shorten >= 1pt,auto]
\node[state] (p1) {%
  $~p_1$
};
\node (left) [left =0.7 of p1]{};
\node[state] (p2) [right= 1.5 of p1] {%
  $~p_2$};
\node[state, fill=gray!60] (env) [above right = .75 and .25 of p1]{$\text{env}$};

\node[state, draw = none, minimum height=.1cm] (p1below) [below= .5 of p1] {%
 };
\node[state, draw = none, minimum height=.1cm] (p2below) [below= .5 of p2] {%
 };

\node (phantom) [below =.3 of p1, draw = none]{\phantom{test}};
\path (env) edge[thick, transform canvas={yshift=0mm}] node[above left, yshift=-1mm] {%
        $\vin_{p_1}$
      }  (p1)
      (env) edge[thick, transform canvas={yshift=0mm}] node[above right, yshift=-1mm] {%
        $\vin_{p_2}$
      }  (p2)
      (p1) edge[thick, transform canvas={yshift=1.5mm}] node[above] {%
        $\vc_{p_1}$
      } (p2)
      (p2) edge[thick, transform canvas={yshift=-1.5mm}] node[below] {%
        $\vc_{p_2}$
      } (p1)
      (p1) edge[thick, transform canvas={yshift=0mm}] node[left] {
      $\vout_{p_1}$
      }(p1below)
      (p2) edge[thick, transform canvas={yshift=0mm}] node[right] {
      $\vout_{p_2}$
      }(p2below)
      ;
\end{tikzpicture}
         }
         \caption{}
         \label{fig:atomic-architecture}
     \end{subfigure}
     \begin{subfigure}[b]{0.49\textwidth}
         \centering
         \resizebox{.9\linewidth}{!}{
         \tikzstyle{state}=[draw, circle, fill=none, minimum width=1cm, 
minimum height = 1.3cm, align=center, thick]

\begin{tikzpicture}[->,>=stealth',shorten >= 1pt,auto]

\node[state] (p0) {%
    $b_0$\\
   $\neg\vout$
};

\node (left) [left = 0.5 of p0]{};

\node[state] (p1) [above right = .1 and 1.5 of p0] {%
  $b_1$\\
  $\vout$
 };
\node[state] (p2) [below right= .1 and 1.5 of p0]{
    $b_2$\\
    $\neg \vout$
    };


\path (left) edge (p0)
      (p0) edge[thick] node[above left, xshift=5pt, yshift=1pt] {%
        $ \vic_0$
      } (p1)
      (p0) edge[thick] node [below, align=center] {%
      $ \vic_1$
      }(p2)
       (p0) edge[loop left, thick, draw=none] node [left] {%
       $\phantom{\LTLtrue}$}
       (p0) 
      (p2) edge[loop right, thick] node [right] {%
      $\LTLtrue$}(p2) 
      (p1) edge[loop right, thick] node [right] {%
      $\LTLtrue$}(p1)
      (p0) edge[loop above, thick] node [above] {%
      $\neg \vic_0 \wedge \neg \vic_1$}(p0)
      ;
\end{tikzpicture}
         }
         \caption{}
         \label{fig:componentspecificationb}
     \end{subfigure}
    \caption{The architecture used for our experiments in (a) where the number outputs, inputs, and communication channels can vary. \Cref{fig:componentspecificationb} shows the implementation of process $b$ for its bit transmission component specification.}
        \label{fig:bit-transmission}
\end{figure}
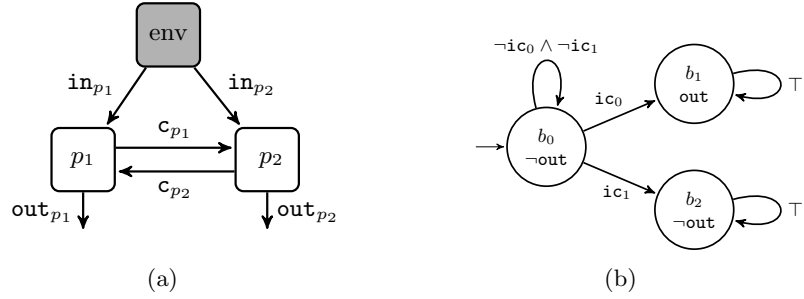
\noindent

The component specification allows us to replace the locality condition (Def.~\ref{def:locality:condition}), which is a hyperproperty, with a trace property. 
Note, however, that the process additionally needs to satisfy the information flow assumption of the other process, which may in general depend on the full set $O_e$ of environment outputs. This would require us to synthesize the process on the full set $O_e$, and to re-introduce the locality condition. In practice, however, the information flow assumption of one process often only depends on the information of the other process. In this case, it suffices to synthesize each process based only on the locally visible environment outputs. 

\Cref{fig:componentspecificationb} shows the implementation of $b$ for its component specification $\langle \varphi_b \rangle$.
In contrast to its hyper implementation (cf. \Cref{fig:hyper-implementation-b}), it does not branch according to $\vin$ and $\vt_p$, but only variables in $IC$.
The specification is encoded as the following LTL formula:\\
\resizebox{\textwidth}{!}{
$\langle \varphi_b \rangle = \big(\LTLglobally \neg\vic_0 \vee \LTLglobally \neg\vic_1)\wedge\LTLfinally \big((\vic_0 \vee \vic_1)\big)
\rightarrow  \big((\LTLfinally \vic_0 \rightarrow \LTLfinally \vout) \wedge (\LTLeventually \vic_1 \rightarrow \LTLglobally \neg \vout)\big)$}

The left hand side of the implication represents the assumption $\psi$, the right hand side specifies the guarantee for each information class.
The composition and decomposition can be performed analogously to the hyper implementations, where we map the value of $\vic$ to the values of the communication variables.
We construct the automata for component specification as follows. 
\begin{enumerate}
    \item By complementing the automaton for the time-bounded distinguishability relation, we obtain an automaton $\mathcal A_{IC}$ that associates each trace over $O_e$ with its information class: i.e., the pair $(v,w)$ of traces over $O_e$ is accepted by the complement automaton iff $(v,w)$ is not in the time-bounded distinguishability relation. 
    \item We obtain the information classes in the following iterative process (under the assumption that the number of information classes is finite):
    \begin{enumerate}
        \item We identify some trace $v$ such that there is a pair $(v,w)$ in the language of $\mathcal{A}_{IC}$;  
        \item for each such trace $v$, we compute an automaton $\mathcal A_{[v]}$ for the information class $[v]$, i.e., an automaton that accepts all traces $v'$ with $(v',w') \in \mathcal{L}(\mathcal{A}_{IC})$ iff $(v,w') \in \mathcal{L}(\mathcal{A}_{IC})$ for all $w'$;
        \item we eliminate all $(v,w)$ with $v \in \mathcal L(\mathcal{R})$ from $\mathcal{A}_{IC}$;
        \item repeat until the language of $\mathcal{A}_{IC}$ is empty.
    \end{enumerate} 
    \item We build an automaton $\mathcal A_{\phi_{p,c}}$ for the relativized specification. The automaton uses universal branching to guess the trace from the information class and applies $\varphi_p$ to
    each guess.
    \item Using the automata $\mathcal A_{[v]}$ for the information classes we build an automaton $\mathcal A_\psi$ for condition $\psi$ from Definition~\ref{def:componentspec}.  
    \item Using $\mathcal A_\psi$ and $\mathcal A_{\phi_{p,c}}$, we build an automaton $\mathcal A_{\varphi_p}$for the component specification.  
\end{enumerate}

\section{Experiments}\label{sec:experiments}
The focus of our experiments is on the performance of the compositional synthesis approach compared to non-compositional synthesis methods for distributed systems. While the time-bounded information flow assumptions and the component specification can be computed automatically by automata constructions, we have, for the purpose of these experiments, built them manually and encoded them as formulas in HyperLTL or LTL, which were then entered to the \textsc{BoSy}/\textsc{BoSyHyper}~\cite{HyperBosy} synthesis tool.
Our experiments are based on the following benchmarks:
\begin{itemize}
	\item \textbf{AC.} \textit{Atomic commit}. The atomic commitment protocol specifies that the output of a local process is set to true iff the observable input and the unobservable inputs are true as well. We only consider one round of communication, the initial input determines all values. The parameter shows how many input variables each process receives, Par. = 1 for the running example.
	\item \textbf{EC.} \textit{Eventual commit}. The atomic commit benchmark extended to eventual inputs - if all inputs (independently of each other) eventually will be true, then there needs to be information flow.
	\item \textbf{SA.} \textit{Send all}. Every input of the sender is relevant for the receiver, so it will eventually be sent if it it set to true. The parameter represents the number of input values and therefore the number of information classes.
\end{itemize}

Table~\ref{tab:experiments} shows the performance of the compositional synthesis approach. 
The column architecture (Arch.) signalizes for each benchmark if the information flow is directional (dir.) or bidirectional (bidir.). Column (Inflow send) indicates the running time for the sending process; where applicable, column  (Inflow rec.) indicates the running time for the synthesis of the process that only receives information. 

We compare the compositional approach to \textsc{BoSyHyper}, based on a standard encoding of distributed synthesis in HyperLTL (Inc. \textsc{BoSy}), and a specialized tool for distributed synthesis~\cite{DistributedBosy} (Distr. \textsc{BoSy}).
All experiments were performed on a MacBook Pro with a 2,8 GHz Intel Quad Core processor and 16 GB of RAM.
The timeout was 30 minutes.

Information flow guided synthesis outperforms the standard approaches, especially for more complex components.
For example, in the atomic commitment benchmark, scaling in the number of inputs does not impact the synthesis of the local processes, while Distr. \textsc{BoSy} eventually times out,  and the running time of Inc. \textsc{BoSy} increases faster than for the information flow synthesis.
For all approaches, the Send All benchmark is the hardest one to solve.
Here, each input that will eventually be set needs to be eventually sent, which leads to non-trivial communication over the shared variables and an increased state space to memorize the individual inputs.
Nevertheless, the information flow guided synthesis outperforms the other approaches and times out with parameter 3 because \textsc{BoSyHyper} cannot cope with the number of states needed.
Synthesizing a receiver that does not satisfy an information flow assumption is close to irrelevant for every benchmark run.
Since these processes are synthesized with local LTL specifications, scaling only in the number of local inputs or information that will eventually be received is easily possible.
Notably, these receivers are compatible with any implementation of the sender, whereas the solutions of the other approaches are only compatible for the same synthesis run.

\begin{table}[t]
\centering
\caption{The results of the experiments with execution times given in seconds. The cell is highlighted if it was faster than the other approaches, where the sum of sender and receiver is taken as reference.}
\resizebox{0.9\textwidth}{!}{
		\begin{tabular}{l | c | c | c | c | c | c }
			Bench.~ & ~Arch.~ & ~Par.~ & ~Inflow send. &~Inflow rec. &  ~Distr.\textsc{BoSy}~  & ~Inc. \textsc{BoSy}~ \\
			 \hline
			AC & dir & 1 & 0.92 & 0.70 & \textbf{1.41} & 2.31 \\
			 & dir & 2 & \textbf{0.36} & \textbf{1.28}    & 2.86 &  2.30\\
			 & dir & 3 & \textbf{0.92} & \textbf{0.68} & 2.46 &  2.55\\
			 & dir & 4 & \textbf{0.92} & \textbf{0.79} & 720.60 & 3.41 \\
			 & dir & 5 &\textbf{0.92} &\textbf{0.68} &  TO & 9.27\\
			 & bidir &  1 &  1.45 & - & \textbf{0.96} & 9.27\\
			 & bidir &  2 &  \textbf{2.49} & - & TO & TO \\
			 & bidir &  3 &  \textbf{79.18}& - & TO & TO\\
			 & bidir &  4 &  TO& - & TO& TO \\ 
			 \hline
			EC & dir &  1 & 0.68 & 1.87 & \textbf{0.92} & 2.556 \\
			 & dir & 2 & 0.94 & 1.85 & \textbf{0.96} &  3.90 \\
			 & dir & 3 & \textbf{202.09} & \textbf{1.78} &  TO & TO\\
			 & dir & 4 & TO & TO &  TO & TO\\
			 & bidir &  1 &   \textbf{3.77} & - & 4.63 & 147.46\\
			 & bidir &  2 &  TO & - & TO & TO\\
\hline
			SA & dir &  1 & 1.31 & 0.92 & 2.21 & \textbf{1.579}\\
			 & dir & 2 &  \textbf{1.78} & \textbf{0.92} & 27.47 & TO\\
			 & dir & 3 &  TO & 1.08 & TO & TO\\
		\end{tabular}
}
	\label{tab:experiments}
\end{table}

\section{Related Work}\label{sec:relatedwork}
Compositional synthesis is often studied in the setting of \emph{complete information}, where all processes have access to all environment outputs~\cite{KupfermanPV06, KuglerS09, FiliotJR10, FinkbeinerP20}. In the following, we focus on compositional approaches for the synthesis of distributed systems, where the processes have incomplete information about the environment outputs. Compositionality has been used to improve distributed synthesis in various domains, including reactive controllers~\cite{DBLP:conf/cav/AlurMT16, DBLP:conf/atva/Hecking-Harbusch19}.
Closest to our approach is assume-guarantee synthesis~\cite{DBLP:conf/tacas/ChatterjeeH07,DBLP:conf/tacas/BloemCJK15}, which relies on behavioral guarantees of the processs behaviour and assumptions about the behavior of the other processes.
Recently, an extension of assume-guarantee synthesis for distributed systems was proposed \cite{DBLP:journals/tcad/MajumdarMSZ20}, where the assumptions are iteratively refined.
Using a weaker winning condition for synthesis, remorse-free dominance \cite{DammF14} avoids the explicit construction of assumptions and guarantees, resulting in implicit assumptions. 
A recent approach~\cite{DBLP:conf/atva/FinkbeinerP21} uses behavioral guarantees in the form of certificates to guide the synthesis process. Certificates specify partial behaviour of each component and are iteratively synthesized. The fundamental difference between all these approaches to the current work is that the assumptions are behavioral. To the best of our knowledge, this is the first synthesis approach based on information-flow assumptions.
While there is a rich body of work on the verification of information-flow properties (cf. \cite{DBLP:conf/vmcai/DimitrovaFKRS12, DBLP:conf/cav/FinkbeinerRS15 , DBLP:conf/csfw/YasuokaT10}), and the synthesis from information-flow properties and other hyperproperties has also been studied before (cf. \cite{HyperBosy}), the idea of utilizing hyperproperties as assumptions for compositional synthesis of distributed systems is new.

\section{Conclusion}\label{sec:conclusion}
The approach of the paper provides the foundation for a new class of distributed synthesis algorithms, where the assumptions refer to the flow of information and are represented as hyperproperties. In many situations, necessary information flow assumptions exist even if there are no necessary behavioral assumptions.  There are at least two major directions for future work. The first direction concerns the insight that compositional synthesis profits from the generality of hyperproperties; at the same time, synthesis from hyperproperties is much more challenging than synthesis from trace properties. To address this issue, we have introduced the more practical method in Section~\ref{sec:practical}, which replaces locality, a hyperproperty, with the component specification, a trace property. However, this method is limited to information flow assumptions that refer to a finite amount of information. It is very common that the required amount of information is infinite in the sense that the same type of information must be transmitted again and again. We conjecture that our method can be extended to such situations. 

A second major direction is the extension to distributed systems with more than two processes. The two-process case has the advantage that the assumptions of one process must be guaranteed by the other. With more than two processes, the localization of the assumptions becomes more difficult or even impossible, if multiple processes have (partial) access to the required information. 

\bibliographystyle{splncs04} 
\bibliography{Paper.bib}

\end{document}